\documentclass[onecolumn,preprint,preprintnumbers,11pt]{article} 

\usepackage{graphicx}
\usepackage{bm}
\usepackage{amsmath,hyperref}
\usepackage{amssymb}
\usepackage{amsfonts}
\usepackage{mathtools}
\usepackage{fancyhdr}
\usepackage{slashed}
\usepackage{latexsym,epsfig,bbm}
\usepackage{amsthm}
\usepackage{bbm}
\usepackage{xcolor}

\usepackage{color}
\definecolor{Pr}{rgb}{0.4,0.3,0.9}

\definecolor{Yq}{rgb}{1.0,0.0,0.0}

\definecolor{Cc}{rgb}{0.9,0.6,0.3}

\usepackage{mleftright}
\usepackage{algorithm}
\usepackage[noend]{algpseudocode}

\usepackage{epigraph}
\usepackage{enumerate}
\usepackage[shortlabels]{enumitem}

\usepackage{fullpage} 

\makeatletter
\newcommand{\ALOOP}[1]{\ALC@it\algorithmicloop\ #1%
  \begin{ALC@loop}}
\newcommand{\ENDALOOP}{\end{ALC@loop}\ALC@it\algorithmicendloop}

\makeatother

\newcommand{\ket}[1]{|{#1}\rangle}

\usepackage{xcolor}
\definecolor{blue}{rgb}{0,0.2,1}

\definecolor{red}{rgb}{0.9,0,0}

\newcommand{\Ord}[1]{\mathcal{O}\mleft( #1 \mright)}
\newcommand{\tOrd}[1]{\tilde{\mathcal{O}}\mleft( #1 \mright)}
\newcommand{\D}[1]{\mathcal{D}}

\newcommand{\Tht}[1]{\Theta\mleft( #1 \mright)}

\newcommand{\eo}{\mathrm{EO}}

\theoremstyle{plain}

\newtheorem{problem}{Problem}
\newtheorem{theorem}{Theorem}

\newtheorem{lemma}{Lemma}

\newtheorem{defn}{Definition}
\newtheorem{assm}{Assumption}

\def\be{\begin{eqnarray}}
\def\ee{\end{eqnarray}}

\newcommand*{\email}[1]{\href{mailto:#1}{\nolinkurl{#1}} } 

\title{Robust quantum minimum finding with an application to hypothesis selection}
\author{Yihui Quek\thanks{Stanford University,
\email{yquek@stanford.edu}}
\and Clement Canonne\thanks{IBM Research Almaden.
\email{clement.canonne@ibm.com}}
\and Patrick Rebentrost\thanks{Centre for Quantum Technologies, National University of Singapore,
\email{cqtfpr@nus.edu.sg}}
}
\date{\today}

\begin{document}

\maketitle

\begin{abstract}
We consider the problem of finding the minimum element in a list of length $N$ using a noisy comparator. The noise is modelled as follows: given two elements to compare, if the values of the elements differ by at least $\alpha$ by some metric defined on the elements, then the comparison will be made correctly; if the values of the elements are closer than $\alpha$, the outcome of the comparison is not subject to any guarantees. We demonstrate a quantum algorithm for noisy quantum minimum-finding that preserves the quadratic speedup of the noiseless case: our algorithm runs in time $\tOrd{\sqrt{N (1+\Delta)}}$, where $\Delta$ is an upper-bound on the number of elements within the interval $\alpha$, and outputs a good approximation of the true minimum with high probability.
Our noisy comparator model is motivated by the problem of hypothesis selection, where given a set of $N$ known candidate probability distributions and samples from an unknown target distribution, one seeks to output some candidate distribution $\Ord{\varepsilon}$-close to the unknown target. 
Much work on the classical front has been devoted to speeding up the run time of classical hypothesis selection from $\Ord{N^2}$ to $\Ord{N}$, in part by using statistical primitives such as the Scheff\'{e} test. Assuming a quantum oracle generalization of the classical data access and applying our noisy quantum minimum-finding algorithm, we take this run time into the sublinear regime.  The final expected run time is $\tOrd{\sqrt{N(1+\Delta)}}$, with the same $\Ord{\log N}$ sample complexity from the unknown distribution as the classical algorithm.
We expect robust quantum minimum-finding to be a useful building block for algorithms in situations where the comparator (which may be another quantum or classical algorithm) is resolution-limited or subject to some uncertainty. This model of metrological errors is fairly general, arising not only in physics and statistics but also in settings where human bias constitutes the noisy comparator, such as in experimental psychology and sporting tournaments.
\end{abstract}

\section{Introduction}
\subsection{Background and motivation}

Optimization is fundamental for a wide range of applications from artificial intelligence and statistical learning to finance and sports. 
An important class of discrete optimization problems are problems where one is given a set of elements and one seeks the minimum (or maximum) element based on some comparison or preference relationship. 
For many sets, such as the integers $\{1,\cdots,N\}$, there is an obvious ordering. For many other examples arising in practical applications such ordering is subject to faults and imprecision. An example one may think of are preference relationships of participants in an economic exchange, which are subject to randomness and human foibles, in other words less consistent in practice than the axiomatic definition of von Neumann and Morgenstern in the 1940s \cite{vNeumann1947}. 
In another example, closer to the topic of this work, the preference or comparison relationship may arise from a computation based on data, that is, it inherently contains statistical fluctuations/sampling errors. An important question is whether the optimization problem of finding the minimum (or maximum) can still be solved using this faulty, imprecise, or noisy comparison operation. 
We consider the problem of finding the minimum in a list of $N$ elements, each holding a value from a set with well-defined comparison relationship and distance metric. For simplicity we assume these values are all distinct.
Classically, the primitive for minimum-finding is the operation of pairwise comparisons between list elements, each invocation of which we call a `classical query'. We assume the primitive is implemented by an oracle $O\colon[N] \times [N] \to [N]$ that takes as input the indices of two elements to compare and outputs the index representing the smaller element. Using a standard generalization, a quantum version of this oracle performs the same operation but with the additional assumption that coherent superpositions over input indices are accepted and the oracle is a unitary circuit and hence reversible. The goal for a minimum-finding algorithm, quantum or classical, is then to minimize the number of queries. It is clear that $N-1$ queries are necessary and sufficient for classical minimum finding; D\"{u}rr and H{\o}yer showed \cite{durr1996quantum} that a quantum computer requires only $\Ord{\sqrt{N}}$ queries. 

What if the comparison is noisy, adversarial or imprecise? We consider the following model of noise: the comparison is correct if the elements being compared are more than a distance $\alpha$ apart; otherwise, the comparison is arbitrary and possibly even adversarial. We give a quantum algorithm for finding the minimum with such a noisy comparator, given an assumption on the density of elements in the list. We show that on a quantum computer, the quadratic speed-up in query complexity is preserved even in this noisy setting. 
This model of noisy comparisons is motivated by the problem of hypothesis selection, itself an important subroutine for many distribution learning problems. We use the Scheff\'{e} test, a well-known statistical primitive that compares two hypotheses in terms of $\ell_1$ distance to an unknown distribution, and the result of this comparison is subject to noise of the form described.

\subsection{Notation}
We denote by $[N]$ the set $\{ 1,\dots, N\}$ and by $[N]'$ the set $\{ 0,1,\dots, N\}$. We denote $\mathbbm Z_+$ the set of all positive integers. Unless otherwise specified, all logarithms are in base 2. We use $\tOrd{}$ to hide poly-logarithmic factors in any of the variables. Define the $\ell_1$-norm of a vector $v \in \mathbbm R^N$ as $\Vert v \Vert_1 = \sum_{j=1}^N \vert v_j\vert$.

\subsection{Problem statement and results}
We concern ourselves mainly with the following problem:

\begin{problem}[Robust minimum finding]
Given a list of $N$ distinct elements from a set $\{x_i\}_{i=1}^N$ with well-defined but inaccessible distance metric $d(x_i,x_j)$, find the minimum using an imprecise or noisy pairwise comparator between elements. The comparator is \emph{imprecise} in the sense that, for some fixed $\alpha\geq 0$, if $d(x_i,x_j)<\alpha$, then the result of the comparison between $x_i,x_j$ can be $>$ or $<$. Otherwise, the result of the comparison is correct.
\end{problem}

Our motivation for studying this problem comes from the following closely-related problem:
\begin{problem}[Hypothesis selection]
\label{prob:hypo}
Given samples from an unknown distribution $q$ and access (to be defined) to a known set $\mathcal{P} = \{p_1,\ldots p_N\}$ of $N$ `hypothesis' distributions, output the distribution $\hat{p} \in \mathcal{P}$ with the smallest $\ell_1$-distance to the unknown distribution. 
\end{problem}
The relation comes from our choice to use a statistical primitive, known as the \emph{Scheff\'{e} test}, as a noisy comparator that decides which of two hypotheses is `closer' to the unknown; that is, the noisy comparison is on the metric induced by $\ell_1$-distance to the unknown distribution.
The Scheff\'{e} test is an efficient way to obtain such a noisy comparator, requiring only a total of $\log(N)/\varepsilon^2$ samples in total, independent of the (possibly infinite) domain size. In contrast, computing the $\ell_1$-distance directly for every hypothesis would require a number of samples scaling polynomially in the domain size and inverse-polynomially (with a large exponent) in the approximation error \cite{Bravyi_2011}, or (classically) a near-linear dependence on the domain size and inverse-quadratic dependence on the approximation error~\cite{VV11,JiaoHW18}. (Note that a maximum likelihood estimator for the $\ell_1$ distance would require a linear dependence on the domain size.)

Since the comparator is noisy, we can only hope to approximate the minimum in both cases. As such, we give algorithms for the above problems that have bounded error probability. We assume an upper-bound $\Delta$ on the number of elements within an interval of length $\alpha$ of any particular element. Appropriately rescaling $\Delta$ allows one to extend this assumption to unit balls in any arbitrary metric space. This assumption is justified later. 

At the core of our work is the quantum minimum finding algorithm by D\"{u}rr-H{\o}yer \cite{durr1996quantum}. We extend this algorithm in two ways. We parameterize the total compute time in terms of the number of pivot changes and also parameterize the finding of not a single but multiple `good-enough' minima. In combination with classical minimum-selection algorithms, we then show that the D\"{u}rr-H{\o}yer algorithm can be applied to the noisy setting with modified run time and correctness guarantees, which we call \textit{robust quantum minimum finding}.
The following theorem summarizes this result. 
\begin{theorem}
Given quantum access to an imprecise comparator, robust minimum-finding can be done in time
$\tOrd{\sqrt{N(1+\Delta)}}$ on a quantum computer, outputting with high probability an element within $2\alpha$ of the true minimum. 
\end{theorem}
The application of this theorem to hypothesis selection leads to a sublinear time algorithm. The main result is as follows. 
\begin{theorem}
For any $\varepsilon > 0$, given $\Ord{\frac{1}{\varepsilon^2}\log N}$ samples from the unknown distribution $q$ and  access a quantum circuit implementing the Scheff\'{e} test on the hyptheses in $\mathcal P$, hypothesis selection can be performed on a quantum computer with a run time of  $\tOrd{\sqrt{N(1+\Delta)}}$, outputting with probability at least $1-\delta$ a hypothesis $\hat{p}\in \mathcal P$ that satisfies
$$
    \lVert \hat{p} - q \rVert_1 \leq 9\min_{p\in \mathcal{P}} \lVert p - q \rVert_1 + \varepsilon.
$$
\end{theorem}

\subsection{Related work}

Three veins of research converge in this work: classical search with imprecise comparators, classical hypothesis selection, and quantum algorithms based on variants of Grover's search. 

The convergence of the first two has been known since 2014, when Acharya et al.~\cite{acharya2016maximum} pointed out that the Scheff\'{e} test functions act as a noisy comparator on the space of probability distributions with an $\ell_1$ distance metric. This Scheff\'e test has proven itself a fundamental building block in many statistical tasks, and underlies many density estimation results~\cite{DevroyeLugosi2001,Tournament,DaskalakisDS14,Diakonikolas16}. Improving on the efficiency of this key subroutine, as well as extending it to new settings such as privacy-aware inference~\cite{Bun0SW19,GopiKKNWZ20}, is therefore of significant importance, and as such has recently received a lot of attention~\cite{acharya2016maximum,Tournament}. Our work follows this line of research, and can be seen as porting this crucial statistical building block to the quantum toolkit. 

On the quantum front, many works have considered quantum search in various faulty models, which essentially can be grouped into `coherent' \cite{Long2000,Hoyer_Noisygrover2003,Shapira2003,Iwama2005} 
and `incoherent' \cite{Shenvi2003,Regev_noquantumspeedup2008}
settings. H{\o}yer, Mosca, and de Wolf~\cite{Hoyer_Noisygrover2003} consider the coherent setting where the oracle outputs the correct value in most cases but with certain failure probability outputs the wrong result. A simple majority voting algorithm yields a $\Ord{\sqrt N \log N}$ run time and the work uses an improved recursive interleaving technique to obtain  
$\Ord{\sqrt N}$. Regev and Schiff~\cite{Regev_noquantumspeedup2008} consider incoherent errors in the sense that the output of the oracle is a statistical mixture of correct output and identity operation (i.e., without coherence between the two outputs.) In this setting, the Grover speedup vanishes and one obtains a run time of $\Tht{N}$ as in the classical naive search algorithm. 

\subsection{Organization}
The paper is organized as follows. In Section \ref{sec:prelims} we define notation, explain the noise model and state some core subroutines. In Section \ref{sec:qmf} we review the quantum minimum finding algorithm. In Section \ref{sec:pivotqmf} we present the workhorse for the subsequent sections, the pivot-counting quantum minimum finding algorithm. In Section \ref{sec:RQMF} we present our two main quantum algorithms for finding a good approximation to the minimum in a list. The first algorithm, entitled `Repeated Pivot-Counting Minimum Finding', is a simple-to-analyze algorithm and the second algorithm, `Robust Quantum Minimum Finding' is a version with better guarantees but worse run-time. Finally, the application of our algorithm to the problem of hypothesis selection is discussed in Section \ref{sec:app}. 

\section{Minimum finding with a noisy comparator}
\label{sec:prelims}
We study minimum-finding in a list $\mathcal L := \{x_i\}_{i=1}^N$ of elements $x_i$, with well-defined but unknown ordering relationship $\text{Comp}(i,j)$ and distance metric $d( x_i, x_j)$.
The true comparator never errs, and so we shall sometimes refer to this as the `noiseless' case. That is, 
\begin{equation}\label{eq:noiselesscomparator}
\text{Comp}(i,j) = \arg \min \left\{x_{i}, x_{j}\right\} \qquad \forall i,j\in [N].
\end{equation}
The true comparator defines the rank of the elements in the set. We assume that all elements are distinct. 

\begin{defn}[Rank of an element]
The rank of an element is defined with regard to the values of the input list and their true ordering relationship. Formally, ${\rm rank}(j) := 1+ \vert \{ i: \text{Comp}(i,j) = i\} \vert$. The element of rank $1$ is the minimum-value element.
\end{defn}

In the practical situation we are only given a noisy pairwise comparator that acts upon two input indices $i, j$ as
\begin{equation}\label{eq:noisycomparator}
\text{NoisyComp}(i,j) = \left\{\begin{array}{ll}
\text{Comp}(i,j)
& {\text { if } d(x_{i},x_{j}) >1} \\
{\text {unknown (possibly adversarial)}} & {\text { otherwise.}}
\end{array}\right.
\end{equation}
The noise model covers any comparator that cannot correctly distinguish inputs distance $\alpha$ apart (we set $\alpha=1$). The noise here may appear intractable at first as we work without any probabilistic guarantees on the comparator's output when the elements are within distance $1$ from each other (informally, we say they are `close'). However, we make a simplifying assumption: we assume that there is an upper limit on the number of list elements within any interval of length $1$ from any particular element: namely, at most $\Delta$ elements can be contained in such an interval.

\begin{defn}[Fudge zone of element $j$\label{def:fudge}]
The \emph{fudge zone of $x_j \in \mathcal L$}, $\text{Fudge}(j)$, is the set of elements that the noisy comparator may make a mistake on when comparing with $x_j$, not including $x_j$ itself. That is, $\text{Fudge}(j) := \{x \in \mathcal L \setminus \{x_j\}: d(x,x_j) \leq 1\}. $
\end{defn}

\begin{assm}\label{assm:1}
There exists $\Delta \in [N]'$ such that at most $2\Delta$ elements are contained in the fudge zone of any element in the list. Formally, for all $x_j \in \mathcal L$ we have $\vert \text{Fudge}(j) \vert \leq 2 \Delta$.
\end{assm}

For instance, in the hypothesis selection setting, when the set of candidate hypotheses is obtained by a cover of the target distribution class (i.e., a gridding), as is the case in many applications, Assumption \ref{assm:1} will hold, with $\Delta$ a function of the granularity of the cover (see, e.g., \cite[Chapter 7]{DevroyeLugosi2001} and~\cite[Chapter 6.5]{Diakonikolas16}). 

Assumption \ref{assm:1} gives us a measure of control on the density of elements. This control allows our conclusions to hold even when the comparator's output on two `close-by' elements is chosen adversarially. In fact, they hold even in the presence of an \emph{adaptive} adversary; that is, one who need not decide on the outcomes of all possible comparisons beforehand, but may do so taking into account the outcomes of all previous comparisons made by the algorithm. \footnote{As in \cite{acharya2016maximum}, we stipulate that once the outcome of a comparison between two elements is decided, it should be fixed over the course of the algorithm. An outcome is `decided' if that outcome is necessary for the comparator to produce its output on some query.}

\begin{defn}[Adversarial/adaptive adversarial comparator]
Taking our cue from \cite{acharya2016maximum}, 
\begin{enumerate}
\item An adversarial comparator has complete knowledge of the minimum-finding algorithm and the input list $\mathcal{L}$, but must fix its comparison output between every pair of elements before the algorithm starts.
\item An adaptive adversarial comparator has complete knowledge of the minimum-finding algorithm and the input list $\mathcal{L}$; and the outcome of all previous queries made by the algorithm. It is allowed to adaptively choose the output of the comparison on each pair of elements before each query to the classical/quantum comparison oracle.
\end{enumerate}
\end{defn}

We will use as our main classical primitive a randomized algorithm from \cite{acharya2016maximum} for approximate classical minimum-finding with a noisy comparator, known as \textsc{COMB}.

\begin{defn}[$t$-approximation]
An element $y \in \mathcal{L}$ is a $t$-approximation of the true minimum $y^\ast$ if it satisfies $d(y,y^\ast)<t$.
\end{defn}

\begin{theorem}[Run time and accuracy guarantees of \textsc{COMB} (Theorem 15 of \cite{acharya2016maximum})\label{thm:comb}]
There exists a classical randomized algorithm, \textsc{COMB}$(\delta, S)$, that outputs a $2$-approximation of the minimum element in the set $S$ with probability at least $1-\delta$, using $\mathcal{O}\left(|S| \log \frac{1}{\delta}\right)$ queries to the noisy comparator. This is true even if the noisy comparator is adversarial (or even adaptively adversarial). 
\end{theorem}

In particular, \textsc{COMB} uses an expected number of queries linear in the size of the input set even in the presence of adversaries. This algorithm is an improvement over the quadratic query complexity of a classical tournament to find the minimum, wherein all pairwise comparisons between list elements are made with the noisy comparator and the element that wins the most comparisons is output. 

\section{Quantum minimum finding}\label{sec:qmf}

We will use as our main quantum primitive the D\"{u}rr-H{\o}yer Quantum Minimum Finding algorithm \cite{durr1996quantum}, which builds on Quantum Grover Search.
We will be mainly concerned with query complexity in the classical and quantum cases, where {\em query} is taken to mean `query to an oracle'. 
As usual, the query complexity is the main factor determining the total `run time' of an algorithm. The run time counts also additional circuitry required by the algorithm, such as preparing a uniform superposition in $\Ord{\log N}$ gates, for example.
To make classical and quantum settings comparable, we encapsulate the classical comparator in a binary-output oracle. We first define a simple binary function based on the output of the noiseless comparator (Equation \ref{eq:noiselesscomparator}). 
\begin{defn}[Noiseless oracle]
\label{oracleNoiseless}
For $i,j \in [N]$, let
\begin{align} \label{eq: noiseless_h_i}
O^{(0)}_{i}(j) &=
\begin{cases}
    &1 \quad \text{if \text{Comp}$(i,j)$ outputs $j$ (i.e. $x_j<x_i$)}\\
    & 0 \quad \text{otherwise},
\end{cases}
\end{align}
where the superscript $(0)$ indicates that this function is for the noise{\em less} comparator. 
\end{defn}
The subscript $i$ is an implicit argument to the function; in our algorithm it will denote the index of a pivot, to which all indices in the list are compared. For the noisy comparator (Equation \ref{eq:noisycomparator}), define the analogous function. 
\begin{defn} 
[Noisy oracle]\label{oracleNoisy}
For $i,j \in [N]$, let
\begin{align} \label{eq: O_i}
O_{i}(j) &=
\begin{cases}
    &1 \quad \text{if \text{NoisyComp}$(i,j)$ outputs $j$ \text{(i.e. NoisyComp thinks $j<i$)}}\\
    & 0 \quad \text{otherwise}.
\end{cases}
\end{align}
\end{defn}

We now apply the standard procedure (for instance given in \cite{Bravyi_2011}, whose notation we follow) for defining equivalent quantum oracles $\hat{O}^{(0)}$ and $\hat{O}$ given the classical ones, $O^{(0)}$ and $O$: we transform $O^{(0)}$ and $O$ into a reversible form and allow it to accept coherent superpositions of queries. In the following paragraph, when we refer to the oracle \(\hat{O}\), the same is true for $\hat{O}^{(0)}$, since the transformation is the same for both. 

The quantum oracle \(\hat{O}\) is a unitary operator acting on a Hilbert space \(\mathbb{C}^{N} \otimes \mathbb{C}^{N} \otimes \mathbb{C}^{2}\) equipped with a standard basis
\(\{\ket{i}\otimes \ket{j} \otimes \ket{k}\}, i,j \in[N], k \in\{0,1\}\) such that
\begin{equation}\label{eq:O2}
\hat{O}|i\rangle\otimes  |j\rangle \otimes|k\rangle=|i\rangle\otimes  |j\rangle \otimes| k \oplus O_i(j)\rangle
\end{equation}
In other words, querying \(\hat{O}\) on a basis vector \(|i\rangle\otimes  |j\rangle \otimes|0\rangle\), one gets the output of the noisy comparator $O_i(j)$ in the last register, while retaining a copy of the input indices in the first two registers to maintain unitarity. Note, however, that in this work we never use superpositions of the pivot index $i$, which can rather be thought of as a classical parameter. It is easy to see that by using $\ket{-} = \frac{\ket{0}-\ket{1}}{\sqrt{2}}$ instead of $\ket{0}$ in the last register, the above oracle $\hat{O}$ flips the phase of a pair $(i,j)$ of basis vectors on which $O_i(j)$ outputs $1$:
\begin{equation}\label{eq:O3}
\hat{O}|i\rangle\otimes  |j\rangle \otimes|-\rangle= (-1)^{O_i(j)}|i\rangle\otimes  |j\rangle \otimes \ket{-}
\end{equation}
We shall sometimes refer to this informally as `marking' the basis vector $\ket{ij} = \ket{i}\otimes \ket{j}$, or, when $i$ is fixed, `marking' $\ket{j}$. We shall count every invocation of operators \(\hat{O}, \hat{O}^{\dagger}\) as a single quantum query. 

The D\"{u}rr-H{\o}yer Quantum Minimum Finding algorithm \cite{durr1996quantum} is split into a quantum search subroutine from \cite{BBHT98j} (which we call \textsc{QSearchWithCutoff}) and an outer loop counting the total run time and changing the pivot if a better candidate for the minimum is found. 
The loop in the algorithm considers the total run time which is a combination of the run time of the steps inside \textsc{QSearchWithCutoff} and the number of times this function is called.
When noiseless comparison oracle $\hat{O}^{(0)}$ is used, we will sometimes refer to this as the `noiseless' algorithm.
D\"{u}rr and H{\o}yer \cite{durr1996quantum} show that a total run time of $T_{\max} = 22.5 \sqrt{N}+ 1.4 \log ^{2} N$ is sufficient to find the minimum with probability $1/2$.

The algorithms are as follows. The first one is the exponential search algorithm from \cite{BBHT98j} with an explicit run time cutoff of $T_{\rm cutoff}$ added. When used as a subroutine in the D\"{u}rr-H{\o}yer algorithm, this argument is needed to enforce the overall run time cutoff at $T_{\max}$. In addition, we define a flag $b$ to control the counting of the $\log N$ step of preparing the uniform superposition as in the D\"{u}rr-H{\o}yer algorithm, which will be set to $0$ by default in our subsequent algorithms. 

\begin{algorithm}[H]
  \caption{\textsc{QSearchWithCutoff}($\hat{O}, y, T_{\rm cutoff},b=0, \text{list} = \mathcal{L}$)}
  \label{algo:qexpcutoff}
  \begin{algorithmic}
  \Require{$\hat{O}, y, T_{\rm cutoff}, b$ (bit for counting logs; by default $0$), input list (by default $\mathcal{L}$)}
  \State $y' \gets$ Unif$[N]$.
  \If{\({O}_y(y')=1,\) } 
  \State {\bf Output:} $(y',0)$.
  \EndIf
  \State $m \gets 1$. 
  \State $\lambda \gets 6/5$.
  \State $T_{\rm search} \gets 0$. \Comment{Current run time}
  \While {$O_y(y') = 0$ {\bf and} $T_{\rm search} \leq T_{\rm cutoff}$}
  \State Initialize a new register of size sufficient to hold \(\sum_{j} \frac{1}{\sqrt{N}}|j\rangle\).
  \State $g \gets {\rm Unif}[m]$.
  \State Apply $g$ iterations of Grover’s algorithm, where each iteration uses two queries to $\hat{O}$.
  \State Measure the system: let $y'$ be the outcome. 
  \State $T_{\rm search} \gets T_{\rm search} + g + b \times \log N$.
  \State $m \gets \min(\lambda m, \sqrt{ N} )$.
  \EndWhile
  \Ensure $(y',T_{\rm search})$
    \end{algorithmic}
\end{algorithm} 
We note that the original paper of D\"{u}rr-H{\o}yer \cite{durr1996quantum} references \cite{BBHT98j} which does not explicitly write the step of uniform sampling then exiting if the element found is a marked one. This step is however necessary to deal with the case that the number of marked elements for a given pivot is more than $3N/4$, as \cite{BBHT98j} make clear in the proof. \footnote{Should it be the case that the number of marked elements $t>3N/4$, the expected number of times this step must be done is constant, because the probability of obtaining a marked element is given by a geometric random variable with success probability $3/4$. Therefore we do not account for it in the analysis of expected run time.}

\begin{algorithm}[H]
  \caption{(Noiseless/D\"{u}rr-H{\o}yer) Quantum Minimum Finding QMF($\hat{O}$) \cite{durr1996quantum}}
  \label{algo:noiseless_qmf}
  \begin{algorithmic}
  \Require{$\hat{O}$}
  \State $T_{\max} \gets 22.5 \sqrt{N}+ 1.4 \log ^{2} N$.
  \State $y \gets$ Unif$[N]$.
  \State $T_{\rm QMF} \gets 0$.
  \Comment{Current run time}
  \State $T_{\rm search} \gets 0$.
   \While {$T_{\rm QMF} \leq T_{\max}$}
  \State $(y',T_{\rm search}) \gets$ \textsc{QSearchWithCutoff}$(\hat{O}^{(0)}, y, T_{\max}-T_{\rm QMF}, b=1)$.
  \State $T_{\rm QMF} \gets T_{\rm QMF} + T_{\rm search}$.
  \If{ ${O}_y(y')=1$ } $y \leftarrow y'$.
  \EndIf
  \EndWhile
  \Ensure $y$
    \end{algorithmic}
\end{algorithm}

\section{Pivot-counting quantum minimum finding}
\label{sec:pivotqmf}

We introduce a new algorithm that can provide guarantees for quantum minimum finding using the noisy oracle $\hat O$. Our algorithm is conceptually very simple and differs from the D\"{u}rr-H{\o}yer Quantum Minimum Finding algorithm in two ways:
\begin{enumerate}
    \item At each iteration, we run Quantum Exponential Search with a different way of cutting off the run time.  This cutoff bounds the run time but introduces the possibility that each \textsc{QSearchWithCutoff} may fail to change the pivot.
    \item We impose a cutoff not on the total run time, but on the number of runs of \textsc{QSearchWithCutoff}. Since the run time of that algorithm is bounded (see Point 1) this cutoff similarly has the effect of making the total run time bounded.
\end{enumerate}

We call our algorithm `Pivot-Counting Quantum Minimum Finding' because of Point 2 above: we count the number of attempts to change the pivot, i.e., the number of runs of \textsc{QSearchWithCutoff}. In the noiseless case this pivot counting is fundamentally not required and provides no benefits $-$ the original algorithm already provides run time and correctness guarantees. The noisy oracle is somewhat less `well-behaved': for instance, no longer is a rank-decrease guaranteed upon every successful pivot change. Here, the pivot counting is actually the key in our proof to obtaining both run time and correctness guarantees, since we are able to upper bound on the \textit{expected} rank improvement per successful pivot change in the presence of the noisy oracle. 

\begin{algorithm}[H]
  \caption{Pivot-Counting Quantum Minimum Finding \textsc{PivotQMF}($\hat{O}$, $\Delta$, $N_{\rm trials}$)}
  \label{alg:PCQMF}
  \begin{algorithmic}
  \Require{comparison oracle $\hat{O}, \Delta, N_{\rm trials}$}
  \State $k \gets 0$.
  \State $y \gets \text{Unif}[N]$.
   \For {$k \in [N_{\rm trials}]$}
  \State $(y', 0) \gets$ \textsc{QSearchWithCutoff}$(\hat{O}, y, 9\sqrt{\frac{N}{1+\Delta}})$.
  \If{ ${O}_y(y')=1$ } $y\gets y'$.
  \EndIf
  \EndFor
  \Ensure $y$
    \end{algorithmic}
\end{algorithm}

We say a few words about Algorithm~\ref{algo:qexpcutoff}, \textsc{QSearchWithCutoff}, at our chosen run time cutoff and with our noisy oracle. First, we remind the reader of the following fact about the infinite-time version of this algorithm (i.e. $T_{\rm cutoff} \to \infty$) with the noiseless oracle. 

\begin{lemma}[From~\cite{BBHT98j}: expected number of steps of Quantum Exponential Search\label{lem:timepivotchange}]
Let the current pivot be of rank $r>1$, i.e., not the minimum element. Quantum Exponential Search, which is also \textsc{QSearchWithCutoff}$(\hat{O}^{(0)}, y, \infty)$, succeeds in finding a marked element after an expected number of Grover iterations of $ \frac{9}{2} \sqrt \frac{N}{r-1}$.
\end{lemma}
We can provide an equivalent statement for the noisy oracle. Let $\Delta$ and $\hat{O}$ be as given by the noise model in Assumption \ref{assm:1}.
\begin{lemma} 
[Accuracy and run time guarantees of \textsc{QSearchWithCutoff} with noisy oracle and with pivot of rank $>\Delta$]
\label{lemma:PCQMFone-stepguarantee}
Let $\Delta \in [N]'$ and the current pivot $y$ be of rank $r$ such that $N\geq r >\Delta$. Then
\begin{enumerate}
\item \textsc{QSearchWithCutoff}$(\hat{O}, y, \infty)$, finds a marked element at an expected number of Grover iterations of at most $ \frac{9}{2} \sqrt \frac{N}{r-\Delta}$. For $r\geq 2\Delta +1$, we may thus upper bound the expected number of steps as $\frac{9}{2} \sqrt \frac{N}{1 + \Delta}$.
\begin{proof}
If the pivot is at rank $r$, the noisy oracle marks a number of elements that is at least $r - \Delta$ and at most $r + \Delta$. Hence from Lemma~\ref{lem:timepivotchange}, the expected steps until a marked element is found can be upper bounded by the steps it takes to do the search on $r-\Delta$ marked elements.
\end{proof}
\item \textsc{QSearchWithCutoff}$(\hat{O}, y, 9\sqrt{\frac{N}{r- \Delta}} )$ succeeds in finding a marked element with probability at least $\frac{1}{2}$. 
\begin{proof}
Follows from Markov's inequality and running for twice the expected number of steps.
\end{proof}
\end{enumerate}
\end{lemma}
In particular, the second item of Lemma~\ref{lemma:PCQMFone-stepguarantee} shows that there is some probability that each run of \textsc{QSearchWithCutoff} will not find a marked element, and hence will not succeed in the pivot change. Every run of \textsc{QSearchWithCutoff} can thus be viewed as a Bernoulli trial that succeeds if the pivot is changed (i.e.  $O_y(y') =1$) or fails otherwise. Accordingly, we define the notions of `attempted' and `successful' pivot changes.

\begin{defn}[Attempted pivot change] \label{def:pivotchangeattempt}
An `attempted pivot change' is a single run of 
$$
\text{\textsc{QSearchWithCutoff}}\left (\hat{O}, y, 9\sqrt{\frac{N}{1+ \Delta}}\right).
$$
Sometimes we refer to this as a `trial'. 
\end{defn}

\begin{defn}[Successful pivot change] \label{def:pivotchange}
A `successful pivot change' is an attempted pivot change that outputs $y'$ such that $O_y(y') = 1$.
That is, a `successful pivot change' is a single run of Quantum Exponential Search that has successfully output a marked element. 
\end{defn}

The roadmap for the rest of this section is as follows: in subsection \ref{subsec:guarantees_pivots}, we provide guarantees on Algorithm~\ref{alg:PCQMF} in terms of the number of {\em successful} pivot changes. Bearing in mind that we cannot control the number of successful pivot changes directly since it is a random variable, in subsection \ref{subsec:guarantees_trials}, we show how the guarantees of the previous subsection translate to guarantees in terms of the number of {\em attempted} pivot changes. 

\subsection{Guarantees in terms of successful pivot changes \label{subsec:guarantees_pivots}}

Our guarantees are in terms of the ranks of elements, unknown to the algorithm but critical for the proof. Accordingly, we introduce some notation.

\begin{defn}\label{not:allpivots}
For $1 \leq i \leq N_{\rm trials}$, let $t_i$ be the rank of the pivot input into the $i$-th round of \textsc{QSearchWithCutoff} in Algorithm~\ref{alg:PCQMF}.
Then $t_1,\ldots t_{N_{\rm trials}}$ is the sequence of pivots input into the $N_{\rm trials}$ attempted pivot changes.
\end{defn}

Note that if an attempted pivot change is not successful, the next attempted pivot change just re-uses the old pivot. It will be useful to track also the successful pivot changes. 

\begin{defn}\label{not:successfulpivots}
Let there be $N_p$ successful pivot changes out of $N_{\rm trials}$ in Algorithm~\ref{alg:PCQMF}. For $1\leq i \leq N_p$, let $r_i$ be the rank of the pivot input into the $i$-th {\em successful} pivot change (or equivalently for $i>1$, the rank of the pivot output by the $(i-1)$-th successful pivot change).
Then $r_1 = t_1$ and $r_1,\ldots r_{N_p}$ is the sequence of successfully-changed pivots attained by the $N_{\rm trials}$ attempted pivot changes.
\end{defn}

Unlike in the noiseless case, a successful pivot change can also increase the rank of the pivot as the noisy oracle can mark elements above the current pivot. Nevertheless, we will be able to provide some guarantees on the expected next rank once a pivot has changed successfully. Intuitively, the key idea for obtaining a guarantee is as follows: even with the noisy comparator, on expectation we still make positive progress down the ranks for every successful pivot change, as long as the rank of the current pivot is high enough. The lemmas in this subsection convey this intuition. 

The following lemma bounds the `worst-case' expected improvement (i.e. decrease) in rank when a pivot changes successfully. 
\begin{lemma}[Worst-case rank change for $r_i>3(\Delta+1)$]
\label{lemma:worst-case progress}
For $r_i>3(\Delta+1)$, the expected rank $r_{i+1}$ of the next pivot (i.e., the new rank from a successful pivot change) satisfies
$$
    \mathbb{E}[r_{i+1}\mid r_i] \leq \frac{1}{r_i+\Delta} \sum_{s=1}^{r_i+\Delta} s = \frac{r_i+\Delta+1}{2} < \frac{2r_i}{3}.
$$
This upper bound holds over all possible (including adversarial) choices of the noisy comparator. 
\end{lemma}

\begin{proof}
Let $r_i = r$. Recall that $r_{i+1}$ is the rank of the next pivot such that
$O_y(y') = 1$ and the pivot is changed. Consider a noisy comparator that marks all elements from rank $1$ to $r+\Delta$. This implies that at the time that the rank $r$ element was chosen as pivot, the active set (i.e. the set of marked elements) comprised all elements in the set $S = [r+\Delta]$; which can only be the case if the noisy comparator exhibits `worst-case' behavior, outputting that all elements in the fudge zone are less than $r$. The expected rank after one step is 
\begin{equation} \label{eq:worst-case progress}
    \mathbb{E}[r_{i+1}\mid r_i = r] = \frac{1}{r+\Delta} \sum_{s=1}^{r+\Delta} s
    = \frac{r+\Delta+1}{2}.  
\end{equation} 
Here, the expectation is over the uniform output distribution of quantum exponential search (over the set of elements marked less than $r$). In other words, the expectation value is simply the mean rank of all elements in the active set at the time the pivot $r$ was chosen. We now explore how alternative outputs of the noisy comparator could change the active set (i.e. we show that the above behavior of the noisy comparator is indeed the worst-case).

An alternative noisy output of the comparator would not mark some elements in $S$, resulting in an alternative active set that differs from $S$ by the removal of some ranks from $S$. By our assumption on the size of the fudge zone, these ranks that might be missing are the ranks $k$ such that $r-\Delta \leq k \leq r+\Delta$. However, removing ranks above the mean rank of $S$ (that is, $\frac{r+\Delta+1}{2}$) cannot possibly increase the mean rank. Hence when $\Delta$, $r$ are such that the minimum rank that a noisy comparator can remove (namely was $r-\Delta$) satisfies $r- \Delta > \frac{r+\Delta+1}{2}$, i.e., $r>3(\Delta + 1)> 3\Delta +1$, we get that $\frac{r+\Delta+1}{2}$ indeed upper bounds, over all possible noisy comparators, the expected rank after one step.  
\end{proof}

Had a noiseless comparator been used instead, the expected one-step rank change would satisfy $\mathbb{E}[r_{i+1} \mid r_i] = \frac{r_i}{2}$ since the output distribution of quantum exponential search is uniform on the set of marked elements. This is why Quantum Minimum Finding finds the minimum in an expected $\log_{2}N$ iterations. Lemma~\ref{lemma:worst-case progress} conveys the critical intuition that while progress is somewhat slower with a noisy comparator, all is not lost: while $r_i>3(\Delta+1)$, we expect to make quick positive progress down the ranks with each iteration of Quantum Minimum Finding.

Indeed, we may immediately use Lemma~\ref{lemma:worst-case progress} to upper-bound the expected output rank of the final successful pivot if all pivots have rank $\geq 3(\Delta+1)$:

\begin{lemma}[{Bound on expected output rank $\mathbb{E}\!\left[r_{N_p+1}\right]$}]
\label{lemma:finalrankguarantee}
Let $N_p = \left \lceil \frac{\log \left(\frac{N}{4\Delta + 3 }\right)}{\log(3/2)} \right \rceil$. If all pivots have rank above $3(\Delta+1)$, then $\mathbb{E}[r_{N_{p}+1}] \leq  4\Delta + 3.$
\end{lemma}

\begin{proof}
Removing expectations one pivot at a time, we may relate $\mathbb{E}[r_{N_p+1} \mid  r_1]$ to $r_1$ as
\begin{equation}\label{eq:relatingexp}
    \mathbb{E}[r_{N_p+1} \mid r_1] 
    =\mathbb{E}\left [ \mathbb E[r_{N_p+1} \mid r_{N_p}]  |r_1 \right ]  \leq
    \mathbb{E}\left [ \frac{2}{3} r_{N_p} \mid r_1 \right ]  \leq
     \left(\frac{2}{3}\right)^{N_p} r_1 \leq \left(\frac{2}{3}\right)^{N_p} N.
\end{equation}
where the first inequality follows from Lemma~\ref{lemma:worst-case progress}, the second inequality follows from a total of $N_p$ applications of that Lemma, and the last inequality from the fact that the maximum rank is of course  $\leq N$. Thus we need to solve for $k$ such that $\left(\frac{2}{3}\right)^{k} N  \leq 4\Delta + 3$. Indeed we may verify that for all $k> N_p = \left \lceil \frac{\log \left(\frac{N}{4\Delta + 3}\right)}{\log(3/2)} \right \rceil$, this holds.
\end{proof}

Should the condition $r_i>3(\Delta+1)$ not be satisfied, however, the following upper-bound on the expected next rank still holds. 

\begin{lemma}[Worst-case rank change for $r_i \leq 3(\Delta+1)$\label{lemma:rankchangeconverged}]
For $r_i \leq 3(\Delta+1)$, the expected rank $r_{i+1}$ after one successful pivot change satisfies
\begin{equation}
    \mathbb{E}[r_{i+1}\mid r_i] \leq 4 \Delta +3.
\end{equation}
\end{lemma}
\begin{proof}
Take the maximum value of $r_i$  possible by assumption, i.e.,  $3(\Delta+1)$,  and add  $\Delta$ to it, leading to $4\Delta +3$. The maximum possible rank of the successfully-changed pivot is $4\Delta +3$ which also bounds the expectation value. 
\end{proof}

\subsection{Guarantees in terms of attempted pivot changes \label{subsec:guarantees_trials}}
Now we combine these past two Lemmas -- which provide guarantees on expected rank in terms of successful pivot changes -- with the failure probability of each attempted pivot change. The purpose of this section is to reason about the rank of the final pivot output by the $N_{\rm trials}$-th trial. We denote this final pivot, which is also the output of Algorithm~\ref{alg:PCQMF}, as having rank $t_{N_{\rm trials}+1}$. Below, let the sequence of pivots $t_1,\ldots t_{N_{\rm trials}}$ be as in Definition \ref{not:allpivots}.

The intuition of the arguments is as follows: in our sequence of pivots, either there is some pivot $i$, $1\leq i \leq N_{\rm trials}$ whose rank goes below a threshold of $3(\Delta+1)$, or not. We claim that both cases ensure that the final pivot is `pretty-small' on expectation -- i.e., $\mathbb{E}[t_{N_{\rm trials}+1}] \leq 4\Delta+3$. In the first case, that first pivot with rank below threshold is already `pretty-small' and Lemma~\ref{lem:induction} uses induction to show that this property is preserved for all following pivots, including the final one. In the other case, though all pivots are not pretty small, we can still show that on expectation the final output rank is `pretty small'. Our two-step argument relates $N_{\rm trials}$ to the number of successful pivot changes (Lemma~\ref{lem:coinflips}) through a Chernoff bound; and then, in Lemma~\ref{lemma:finalrankguarantee_2}, imports the guarantees of Lemma~\ref{lemma:finalrankguarantee} to bound the expected final output rank. We conclude in Theorem \ref{thm:1PCQMF} by applying Markov's inequality to bound the final output rank and stating the run time of Algorithm~\ref{alg:PCQMF}.

\begin{lemma}[{Bound on expected output rank $\mathbb{E}[t_{N_{\rm trials}+1}]$ 
if intermediate pivot has rank $\leq 3(\Delta+1)$}]
\label{lem:induction}
If there is some pivot whose rank goes below $3(\Delta+1)$, $\mathbb{E}[t_{N_{\rm trials}+1}] \leq 4\Delta + 3$.
\end{lemma}

\begin{proof}
Let $\tau$, $1\leq \tau \leq N_{\rm trials}$ be the minimum index such that $t_{\tau} \leq 3(\Delta+1)$. We claim that if $\tau \leq N_{\rm trials}$, then $\mathbb{E}[t_i \mid t_{\tau}] \leq 4\Delta + 3$ for all $i$ s.t. $\tau\leq i\leq N_{\rm trials}+1$.

To see this, we proceed by induction on $i$. We establish that $\mathbb{E}[t_i\mid t_{\tau}] \leq 4\Delta + 3$ for the base case, which is $i=\tau+1$. Either this pivot change succeeds or not. If the pivot change does not succeed, then the same pivot is used at the next attempt, i.e., $t_{\tau+1}= t_{\tau}$. But by the definition of $\tau$, $t_{\tau} \leq 3(\Delta+1) < 4\Delta + 3$ and so $t_{\tau+1} < 4\Delta + 3$, so clearly $\mathbb{E}[t_{\tau+1}\mid t_{\tau}] \leq 4\Delta + 3$. If the pivot change succeeds, then by Lemma~\ref{lemma:rankchangeconverged}, we know that $\mathbb{E}[t_{\tau+1}\mid t_\tau] \leq 4\Delta + 3$.

Now we prove the inductive hypothesis: suppose $\mathbb{E}[t_i\mid t_{\tau}] \leq 4\Delta + 3$ for all $i$ up to $k$ for $k < N_{\rm trials}$, i.e., such that $\tau+1 \leq i \leq k$; we will show that $\mathbb{E}[t_{k+1}\mid t_{\tau}] \leq 4\Delta + 3$ also. Again consider the two cases. If the $k$th pivot change does not succeed, $t_{k+1} = t_k$ and the same reasoning as above gives the desired statement. If the pivot change succeeds, then there are two subcases: either $\mathbb{E}[t_{k}\mid t_\tau] \leq 3(\Delta+1)$ or $3(\Delta+1) < \mathbb{E}[t_{k}\mid t_\tau] \leq 4\Delta+3$. In the first subcase, we are done, since then $\mathbb{E}[t_{k+1}\mid t_{k}] \leq 4\Delta+3$ by Lemma~\ref{lemma:rankchangeconverged}; in the second subcase, we apply Lemma~\ref{lemma:worst-case progress} to conclude that
\[
    \mathbb{E}[t_{k+1}\mid t_{\tau}] \leq \frac{2}{3}\mathbb{E}[t_{\tau+1}\mid t_{\tau}]
    \leq \frac{2}{3}(4\Delta+3) < 4\Delta + 3\,,
\]
concluding the induction argument.
\end{proof}

The following two Lemmas deal with the case that no intermediate pivot has rank below threshold. We first probabilistically relate the number of attempted pivot changes to the number of successful ones. We use a Chernoff bound to make the probability of failure (i.e. obtaining $<N_p$ successful pivot changes) less than $\frac{1}{N}$ -- this probability will suffice later (in Theorem \ref{thm:1PCQMF}) for a good bound on the expectation value. 

\begin{lemma}[Number of successful pivot changes in $N_{\rm trials}$ trials if all pivots have rank $\geq 3(\Delta+1)$\label{lem:coinflips}]
Let $N_p=\left \lceil \frac{\log \left(N/(4\Delta + 3)\right)}{\log(3/2)} \right \rceil$ be the number of desired successes and $N_{\rm trials} = \lceil 8\max( N_p, 2\ln N )\rceil$. Suppose all pivots in the sequence $t_1,\ldots, t_{N_{\rm trials}}$ are above rank $3(\Delta+1)$. Then with probability at least $1-\frac{1}{N}$, there are $N_p$ successful pivot changes out of $N_{\rm trials}$ attempted ones.
\end{lemma}

\begin{proof}
Since we have assumed $r > 3(\Delta +1)$ while our run time is $\Ord{\sqrt{\frac{N}{\Delta+1}}}$, we may use Lemma~\ref{lemma:PCQMFone-stepguarantee} to see that each attempted pivot change succeeds with probability at least $\frac{1}{2}$. Let $X$ be the Binomial random variable corresponding to the number of successful pivot changes in $N_{\rm trials}$ trials, where the probability of success is at least $\frac{1}{2}$. From a Chernoff bound, for $\kappa \in [0,1]$,
\begin{equation}
    \mathbb{P}[ X \leq (1-\kappa)\cdot N_{\rm trials}/2] \leq e^{-\frac{1}{4}\kappa^2 N_{\rm trials}}.
\end{equation}
Recall that $N_p = \left \lceil \frac{\log \left(N/(4\Delta + 3)\right)}{\log(3/2)} \right \rceil$ is our desired minimum number of successful pivot changes. We would like to choose $N_{\rm trials}$ so that the number of successful pivot changes is at least $N_p$ with probability no less than $1-1/N$, so that we may later use the lemmas from the previous subsection which are based on the number of successful pivot changes. The above bounds leads to choosing
\begin{equation}
N_{\rm trials} = \left\lceil \frac{4}{\kappa^2} \ln N \right\rceil
\end{equation}
for a $\kappa$ which remains to be chosen. From the Chernoff bound, it suffices to choose it such that 
$(1-\kappa)\cdot N_{\rm trials}/2 \geq N_p$ 
or, equivalently, 
$\frac{1-\kappa}{\kappa^2} \geq \alpha := \frac{N_p}{2\ln N}$. This is satisfied whenever
$
    \kappa \leq \frac{\sqrt{4\alpha+1}-1}{2\alpha}
$, and in particular for $\kappa = \min(1/2, 1/(2\sqrt{\alpha}))$ (as $\min(x,\sqrt{x}) \leq \sqrt{4x+1}-1$ for all $x\geq 0$). 
Taking $\kappa^2 = \frac{1}{4}\min(1,\frac{2\ln N}{N_p})$, we obtain
\begin{equation} \label{eq:ntrials} 
N_{\rm trials} = \left\lceil 8\max( N_p, 2\ln N ) \right\rceil.
\end{equation}
In other words, for $N_{\rm trials} \geq \lceil 8\max( N_p, 2\ln N )\rceil$, with probability at least $1-\frac{1}{N}$, there are at least $N_p$ successful attempted pivot changes out of the $N_{\rm trials}$ attempts. 
\end{proof}

\begin{lemma}[Bound on expected output rank 
if all pivots have rank $\geq 3(\Delta+1)$]
\label{lemma:finalrankguarantee_2}
Let $N_p = \left \lceil \frac{\log \left(\frac{N}{4\Delta + 3 }\right)}{\log(3/2)} \right \rceil$ and $N_{\rm trials} = \lceil 8\max( N_p, 2\ln N )\rceil$. If there is no pivot whose rank goes below $3(\Delta+1)$, then with probability at least $1-\frac{1}{N}$, there are $N_p$ successful pivot changes and $\mathbb{E}[t_{N_{\rm trials}+1}] \leq 4\Delta + 3.$
\end{lemma}

\begin{proof}
Lemma~\ref{lem:coinflips} gives that with probability at least $1-\frac{1}{N}$, out of these $N_{\rm trials}$ trials, there are at least $N_p$ successful pivot changes, each of which is above $ 3(\Delta+1)$. In that case, we may import the guarantees of Lemma~\ref{lemma:finalrankguarantee} to conclude that with probability at least $1-\frac{1}{N}$,
\begin{equation}
    \mathbb{E}[t_{N_{\rm trials}+1}] = \mathbb{E}[r_{N_{p}+1}] \leq 4\Delta + 3.
\end{equation}
\end{proof}

\begin{theorem} [Accuracy and run time of \textsc{PivotQMF}] \label{thm:1PCQMF}
Let $N_p = \left \lceil \frac{\log \left(\frac{N}{4\Delta + 3 }\right)}{\log(3/2)} \right \rceil$ and $N_{\rm trials} = \lceil 8\max( N_p, 2\ln N )\rceil$.
With probability $>3/4$,  \textsc{PivotQMF}($\hat{O}, \Delta, N_{\rm trials})$ succeeds in getting an element of rank at most $16(\Delta +1)$ with a number of queries of at most $\Ord{ \sqrt{\frac{N}{\Delta+1}} \log N}$.
\end{theorem}

\begin{proof}
Lemmas \ref{lem:induction} and \ref{lemma:finalrankguarantee_2} yield some measure of control over $\mathbb{E}[t_{N_{\rm trials}}]$: together they tell us that the bad event where we cannot guarantee the number of successful pivot changes (and hence the final rank) occurs with probability at most $\frac{1}{N}$. If so, we trivially bound the rank of the final pivot by $N$. Otherwise, following the directives of these two Lemmas, we upper bound using $\mathbb{E}[t_{N_{\rm trials}}] \leq 4N/3$.

This leads to an overall upper bound as follows: 
\begin{equation} \label{eq:guarantee_on_rank}
    \mathbb{E}[t_{N_{\rm trials}}] \leq \frac{1}{N} (N) + \left(1-\frac{1}{N}\right) (4\Delta+3) \leq 4\Delta+ 4.
\end{equation}
Then, Equation \ref{eq:guarantee_on_rank} combined with Markov's inequality using this upper bound on $\mathbb{E}[t_{N_{\rm trials}}]$ leads to the accuracy guarantee in the Theorem: that with probability at least $3/4$, the actual rank $t_{N_{\rm trials}}$ is less than $4$ times this upper bound on its expectation. 

The run time guarantee follows straightforwardly from the run time of \textsc{QSearchWithCutoff}, in Lemma~\ref{lemma:PCQMFone-stepguarantee}, and the fact that $N_{\rm trials} = \Ord{\log{N}}$.
\end{proof}

It is easy to see that by classically repeating \textsc{PivotQMF} for $\Ord{\log \frac{1}{\delta}}$ times, the success probability for finding the rank $\leq 16 ( \Delta+1)$ can be boosted to $1-\delta$, and this is indeed what we do in the next Section.

\section{Robust Quantum Minimum Finding}\label{sec:RQMF}

In this section, we will introduce two algorithms that build on \textsc{PivotQMF}. The first is Repeated Pivot Quantum Minimum Finding (\textsc{RepeatedPivotQMF}), an intermediate algorithm which simply repeats \textsc{PivotQMF}, and then performs a classical minimum selection on its outputs. The second uses \textsc{RepeatedPivotQMF} as a subroutine for our eventual Robust Quantum Minimum Finding \textsc{RobustQMF} algorithm.

Within the \textsc{RepeatedPivotQMF} algorithm, \textsc{PivotQMF} is used as a workhorse to rapidly obtain the index of some element of low rank (i.e. $\leq 16({\Delta}+1)$). 
The algorithm proceeds in two phases; in the first stage, we use \textsc{PivotQMF} with $N_{\rm trials}= \Ord{\log N}$ to find an element within $16 ({\Delta}+1)$ ranks of the minimum (what we call a ``pretty-small element'') with constant probability; this itself is repeated $\log(2/\delta)$ times to bootstrap the probability of finding at least one pretty-small element to at least $1-\delta$. This yields a pool, $S$, of $\log(2/\delta)$ elements (stored classically) containing with high probability at least one element that has rank $\leq 16(\Delta+1)$. Since the elements are identified by their indices, more processing with the noisy comparator needs to be done on the pool to single out this best element. For this we use the algorithm \textsc{COMB} from \cite{acharya2016maximum} described in Theorem \ref{thm:comb}. The algorithm is given in pseudocode in Algorithm~\ref{algo:simple_robustqmf}.

\begin{algorithm}[H]
  \caption{Repeated Pivot Quantum Minimum Finding \textsc{RepeatedPivotQMF}$(\hat{O},\delta,\Delta)$ }
  \label{algo:simple_robustqmf}
  \begin{algorithmic}
  \Require{$\hat{O},\delta,\Delta$}
  \State $S \leftarrow \varnothing.$
\State \textbf{Stage~I: Finding pretty-small element w.h.p.} 
    \For {$ i = 1, \ldots \log_{4}(2/\delta)$}
\State  $y \gets $ \textsc{PivotQMF}$\left (\hat O, \Delta, \lceil 8\max( \frac{\log(N/(4\Delta + 3)}{\log (3/2)}, 2\ln N )\rceil \right)$
  \State $S \gets S \cup \{y \}$.
  \EndFor
  \State \textbf{Stage~II: Classical minimum selection with noisy comparator}
  \State Perform \textsc{COMB}($\delta/2$, $S$).
  \Ensure Output of \textsc{COMB}($\delta/2$, $S$).
    \end{algorithmic}
\end{algorithm} 
With the lemmas for \textsc{PivotQMF} in hand, we can now prove the following Theorem. 
\begin{theorem}\label{thm:SRQMF}
Algorithm~\ref{algo:simple_robustqmf}, \textsc{RepeatedPivotQMF}$(\hat{O}, \delta,\Delta)$ requires an number of queries of $$\Ord{\sqrt{\frac{N}{1+\Delta}} \log(N)\log{\left (\frac{1}{\delta}\right)} + \log^2\left (\frac{1}{\delta}\right)},$$ and, with probability at least $1-\delta$, outputs an element whose rank is at most $18 \Delta + 16$. This is true even if the noisy comparator is adversarial (or even adaptively adversarial).
\end{theorem}
\begin{proof}[Proof of Theorem \ref{thm:SRQMF}] 
We prove the accuracy guarantee first. Theorem \ref{thm:1PCQMF} shows that with probability at least 3/4, a single round of \textsc{PivotQMF} outputs an element within $16({\Delta}+1)$ of the minimum. Therefore, if we repeat it $\log_{1/4}(\delta/2) = \Ord{\log(1/\delta)}$ times, the probability that none of the repetitions of Stage~I outputs such an element at most $\delta/2$. Hence with probability $1-\delta/2$, the set $S$ has the property that at least one of its elements has rank $\leq 16 ({\Delta}+1)$. Condition on this for Stage~II, which takes in set $S$.

Theorem \ref{thm:comb} guarantees that \textsc{COMB}($\delta', \cdot$) outputs a $2$-approximation of the minimum with probability at least $1-\delta'$, even if the comparator is adversarial. We choose $\delta' = \delta/2$. Conditioned on Stage~I succeeding, the rank of the minimum element in $S$ is at most $16 ({\Delta}+1)$; furthermore, from Assumption \ref{assm:1} there are at most $2\Delta$ elements within an interval of length $2$. Hence with probability $1-\delta/2$ the algorithm outputs an element within $16 ({\Delta}+1)+2\Delta$ ranks of the minimum. A union bound over the probabilities of error of the two stages yields the overall accuracy guarantee.

Now we prove the run time guarantee. The number of queries of \textsc{COMB}($\delta'$,$S$) is $\Ord{|S|\log(1/\delta')}$. Substituting $\delta' = \delta/2, |S| = \log_{4}(2/\delta)$ yields a query complexity of $\Ord{\log(1/\delta)^2}$. Adding this to the run time of Stage~I proven in Theorem \ref{thm:1PCQMF} (multiplied by the $\log(1/\delta)$ iterations) yields the total query complexity. 
\end{proof}

\textsc{RepeatedPivotQMF} quickly finds an element of ``pretty-small" rank, i.e., of rank at most $18 \Delta +16$ with high probability, at the cost of a multiplicative $\log(\frac{1}{\delta})$ run time factor. It is possible, however, to obtain a better guarantee on the approximation of the minimum. If we had a sublist of elements ranking below the ``pretty-small" element that included the minimum, or close-to-minimum element w.h.p., we could run our classical minimum-selection algorithm on this sublist to approximate {\em its} minimum element. \textit{Robust Quantum Minimum Finding}, \textsc{RobustQMF}, simply fixes the output of \textsc{RepeatedPivotQMF} as a pivot and
repeatedly applies \textsc{QSearchWithCutoff} -- always on the same pivot -- to obtain such a sublist. Because the rank of the ``pretty-small" element is small, only linear-in-$(1+\Delta)$ iterations are needed.

However, this brings up a potential problem: we use \textsc{QSearchWithCutoff}$(\hat{O}, Y_{\rm out}, 9\sqrt{\frac{N}{1+\Delta}})$ which has a cutoff compared to the original Quantum Exponential Search. Thus we must be able to bound its probability of finding a marked element, independent of the rank of $Y_{\rm out}$. For $Y_{\rm out}$ of rank $<1+2\Delta$, there may be less than $1+\Delta$ marked elements, but the expected run time of Quantum Exponential Search is $\sqrt{\frac{N}{t}}$ for $t$ being the number of marked elements. Hence our proposed time cutoff actually runs it for {\em less} than its expected run time in that case. We deal with this problem by introducing $k$ {\em dummy elements} that are always marked by the comparator as less than $Y_{\rm out}$. 

Introducing these dummy elements is as simple as appending to our original list of indices, $K$ indices that are `out-of-range' of the comparator. This works because all elements are identified to the comparator by their indices. This also incurs only a mild implementation overhead, requiring a small constant number of additional gates as compared to the original oracle implementation of the comparator. Since we have denoted our original input list as $\mathcal{L}$, we denote the list that has been modified as described above as $\mathcal{L}'$. Our algorithm is as follows.

\begin{algorithm}[H]
  \caption{Robust Quantum Minimum Finding \textsc{RobustQMF}$(\hat{O},\delta,\Delta)$}
  \label{algo:robustqmf}
  \begin{algorithmic}
  \Require{$\hat{O},\delta,\Delta$}
  \State \textbf{Stage~I: Finding a ``pretty-small" element with \textsc{RepeatedPivotQMF}}
  \State $Y_{\rm{}out} \gets$ \textsc{RepeatedPivotQMF}($\hat{O},\delta/2,\Delta$)
  \State \textbf{Stage~II: Finding even smaller elements}
  \State $S \gets \{Y_{\rm out}\}$
  \State $\tilde{T} \leftarrow 19\Delta +16 $
  \For {$ i = 1, \ldots 2 \ln 2 \log(\frac{4}{\delta})\tilde{T}$}
  \State $(y_k,g) \gets$ \textsc{QSearchWithCutoff}$(\hat{O}, Y_{\rm out}, 9\sqrt{\frac{N}{1+\Delta}}, \text{list} = \mathcal{L}')$.
  \If{ \({O}_{Y_{\rm out}}(y_k)=1,\) and $y_k$ is not a dummy element} {\(S \gets S \cup \{y_k \}\)}
  \EndIf
  \EndFor
  \State Remove repeated elements from $S$. 
  \State \textbf{Stage~III: Classical minimum-selection with noisy comparator}
  \State Perform \textsc{COMB}($\delta/4$, $S$).
  \Ensure Output of \textsc{COMB}($\delta/4$, $S$).
  \end{algorithmic}
\end{algorithm} 

To succinctly describe our \textsc{RobustQMF} algorithm: the purpose of Stage~I is to find a $Y_{\rm out}$ with rank $<16(\Delta+1)$ with high probability; the purpose of Stage~II is to find a sublist of $S_{< Y_{\rm out}}$ that includes the minimum element of this set (which we denote as $\min(S_{<Y_{\rm out}})$) with high probability, if such a minimum element exists; the purpose of Stage~III is to find (a good approximation of) this minimum element.

We now prove guarantees for this algorithm. For convenience, we define the following notation which we will use in the next few Lemmas: Let $\tilde{r}:= 16(\Delta +1)$ and $\tilde{T} := 19\Delta +16$. Let $S_{< i}=\{j: O_{i}(j) = 1\}$ be the set of elements that the noisy comparator deems less than element $i$.  

From now on, we assume that $S_{< Y_{\rm out}}$ is not empty for the following reason: $S_{< Y_{\rm out}}$ could only be empty if $Y_{\rm out}$ is $y^\ast$ or in the fudge zone of $y^\ast$, however in that case, all iterations of \textsc{QSearchWithCutoff} will either fail or return a dummy element (as those are the only marked ones) and the for-loop adds no elements to $S$. $S$ contains only $Y_{\rm out}$, so the entire algorithm outputs $Y_{\rm out}$ which is a 1-approximation of $y^\ast$, and our accuracy guarantees still hold. In particular, $S_{< Y_{\rm out}}$ not being empty means its minimum exists.

Since Stage~I is taken care of by the Theorems concerning \textsc{RepeatedPivotQMF}, we now consider Stage~II. We first argue about its key subroutine, the probability of finding a marked element with \textsc{QSearchWithCutoff} on the list $\mathcal{L}'$. More specifically, we compute how many dummy elements are needed to get a desired upper bound on expected number of Grover iterations of infinite quantum exponential search, and then use Markov's inequality to argue what the iterations cutoff should be to obtain a success probability of at least $\frac{1}{2}$.

\begin{lemma}[Required number of dummy elements in input list of Quantum Exponential Search\label{lem:K}]
Let the rank of $Y_{\rm out} <16(\Delta+1)$. For the input list $\mathcal{L}'$ comprising $\mathcal{L}$ appended with $K=2\Delta$ dummy elements, the expected number of Grover iterations of Quantum Exponential Search with pivot $Y_{\rm out}$ on list $\mathcal{L}'$ is at most $\frac{9}{2} \sqrt{\frac{N}{1+\Delta}}$. 
\end{lemma}

\begin{proof}
Let $t = |S_{< Y_{\rm out}}|$, which is also the number of marked elements for pivot $Y_{\rm out}$ in the original list, $\mathcal{L}$. As explained earlier, it suffices to assume that $t\geq 1$. If the rank of $Y_{\rm out}< 16(\Delta+1)$, then $1\leq t \leq 16(\Delta+1)+\Delta$. Also, the number of marked elements in the new list $\mathcal{L}'$ is $t+K$ because additionally all $K$ dummy elements are marked.

The expected number of steps of quantum exponential search is $\frac{9}{2}\sqrt{\frac{\text{no. elements}}{\text{no. marked elements}}}$. Since $|\mathcal{L}'| = N+K$, it suffices for us to solve for $K$ such that $\sqrt{\frac{N+K}{K+t}} < \sqrt{\frac{N}{1+\Delta}}$. Since the square root is a monotonic function and $t\geq 1$, it suffices to find $K$ such that $\frac{N+K}{K+1} < \frac{N}{1+\Delta}$. 
\begin{align}
    \frac{N+K}{K+1} < \frac{N}{1+\Delta} &\leftrightarrow (N+K)(1+\Delta) < NK+N\\
    &\leftrightarrow N(1+\Delta)-N <K(N-(1+\Delta)) \\
    &\leftrightarrow K > \left(\frac{N}{N-(1+\Delta)}\right)\Delta = \left(1 + \frac{1+\Delta}{N-(1+\Delta)} \right)\Delta.
\end{align}
Clearly if $N>2(1+\Delta)$ (which must also be true for the guarantees of Section \ref{sec:pivotqmf} to hold), $\frac{1+\Delta}{N-(1+\Delta)} < 1$ so setting $K =2\Delta$ suffices.
\end{proof}

From Lemma \ref{lem:K}, we may see that running \textsc{QSearchWithCutoff} with a cutoff of $9\sqrt{\frac{N}{1+\Delta}}$, which is at least twice the upper bound on the expected number of iterations, succeeds in finding some element in $S_{< Y_{\rm out}}$ with probability at least $\frac{1}{2}$. For Stage~II to succeed, at least one of its iterations must output $\min(S_{<Y_{\rm out}})$. The following Lemma lower-bounds the number of iterations of \textsc{QSearchWithCutoff} required in Stage~II for this to happen.

\begin{lemma}[Number of iterations required to output $\min(S_{< Y_{\rm out}})$\label{lem:9Tsuffices}]
 If the rank of $Y_{\rm out} <16(\Delta +1)$, then with probability at least $1-\delta$, within $ 2 \ln 2 \log(\frac{1}{\delta})\tilde{T}$ iterations of \textsc{QSearchWithCutoff}($\hat{O}$,$Y_{\rm out}$,$\mathcal{L}'$), at least one outputs $\min(S_{< {Y_{\rm out}}})$.
\end{lemma}

\begin{proof}
From Lemma \ref{lem:K}, every round of \textsc{QSearchWithCutoff} has probability at least $\frac{1}{2}$ of succeeding in finding a marked element. The output distribution of \textsc{QSearchWithCutoff} is uniform over the marked elements. Hence conditioned on it finding a marked element, with probability at least $\frac{1}{\tilde{T}}$, that element is $\min(S_{<{Y_{\rm out}}})$. This is because there are at most $16(\Delta +1)+\Delta$ elements in $S_{<{Y_{\rm out}}}$ and there are $2\Delta$ dummy elements in $\mathcal{L}'$, making for at most $16(\Delta +1)+2\Delta+\Delta =19 \Delta + 16 = \tilde{T}$ marked elements in $\mathcal{L}'$. 
Multiplying the two probabilities,
\begin{equation}\label{eq:failprob}
    \mathbb{P}(\text{\textsc{QSearchWithCutoff} outputs }\min(S_{< {Y_{\rm out}}})) > \frac{1}{2\tilde{T}}.
\end{equation}
Define `failure' as the event that a single trial of \textsc{QSearchWithCutoff} either does not successfully find a marked element or finds a marked element that is not $\min(S_{< {Y_{\rm out}}})$. Equation \ref{eq:failprob} implies that $\mathbb{P}(\rm fail) < 1-\frac{1}{2\tilde{T}}$. We would like to compute $N_{trials}$ so that the probability of all trials failing is upper bounded, \begin{equation}
    \mathbb{P}(\text{all $N_{trials}$ fail}) < \left(1-\frac{1}{2\tilde{T}}\right)^{N_{\rm trials}}< \delta.
\end{equation}
Using the lower bound $\frac{1}{\ln(1+z)} \leq \frac{2+z}{2z} \,\, \forall z>0$, we obtain
\begin{equation}
    N_{\rm trials} = \log_{\frac{2\tilde{T}}{2\tilde{T}-1}} \left(\frac{1}{\delta}\right) = \frac{\log_2 \left(\frac{1}{\delta}\right) \ln 2}{\ln (1 +\frac{1}{2\tilde{T}-1})} \leq 2 \ln 2 \log \left(\frac{1}{\delta}\right) \tilde{T}.
\end{equation}
\end{proof}

\begin{theorem}[Accuracy and run time guarantees of \textsc{RobustQMF}]\label{thm:RQMF_1}
\textsc{RobustQMF}($\hat{O},\delta, \Delta$), see Algorithm~\ref{algo:robustqmf}, outputs an element $y$ that is a $2$-approximation of the minimum element, $y^*$, with probability at least $1-\delta$
at a query complexity of $\Ord{\sqrt{\frac{N}{1+\Delta}}\tilde{R}\log(1/\delta)+ (1+\Delta)\log(1/\delta) + \log(1/\delta)^2}$ where $\tilde{R}:=\max{\left(\Ord{\log\left(N\right)}, \Ord{\Delta +1}\right)}$. This is true even if the noisy comparator is adversarial (or even adaptively adversarial).
\end{theorem}

\begin{proof}
We prove the accuracy guarantee first. Theorem \ref{thm:SRQMF} tells us that w.p. at least $1-\delta/2$, the output of Stage~I, $Y_{\rm out}$, is of rank at most $16(\Delta +1)$, and is stored classically. Condition on this. For Stage~II, Lemma \ref{lem:9Tsuffices} tells us that the output of Stage~II, which is also the set that goes into COMB, contains $\min(S_{< Y_{\rm out}})$ with probability at least $1-\delta/4$. Condition on this too.

Now, either $Y_{\rm out}$ is in Fudge$(1)$ (the fudge zone of the rank-1 element, $y^\ast$) or not. If $Y_{\rm out}$ is not in Fudge$(1)$, then $y^\ast$ is not in Fudge($Y_{\rm out}$). 
In that case, $S_{< Y_{\rm out}}$ must contain $y^\ast$, $\min(S_{< Y_{\rm out}})$ must be $y^\ast$, and in particular, $y^\ast$ is inside the set that goes inside COMB. 
Referring to the accuracy guarantees of COMB in Theorem \ref{thm:comb}, the final output of the algorithm is a $2$-approximation of $y^\ast$ with probability at least $1-\delta/4$. 
Otherwise, if $Y_{\rm out}$ is in Fudge$(1)$ or $y^\ast$ itself, then the largest possible element in $S_{< Y_{\rm out}}$ is the largest possible element $e$ in Fudge($Y_{\rm out}$). By Definition \ref{def:fudge} of the fudge zone, $d(e,Y_{\rm out}) \leq 1$, and by the assumption that $d(Y_{\rm out},y^\ast) \leq 1$ and the triangle inequality, we may conclude that $d(e,y^\ast) \leq 2$. 
Since the output of COMB must be in the set $S_{< Y_{\rm out}}$, the output of COMB, $y$, must, again, satisfy $d(y,y^\ast) \leq 2$. 

From a union bound, the total failure probability of the algorithm is $\delta/2+\delta/4 +\delta/4 = \delta$ and if the algorithm does not fail, the guarantee on the rank of the output element is as stated. 

Next, we prove the run time guarantee. The query complexity of Stage I follows from Theorem~\ref{thm:SRQMF} as $\Ord{\sqrt{\frac{N}{1+\Delta}}\log(N)\log(1/\delta) + \log(1/\delta)^2}$. The run time bottleneck in Stage~II is iterating \textsc{QSearchWithCutoff}, so we will account only for the iterations\footnote{We can use an efficient data structure such as a red-black tree to remove repeated elements.}. The total number of iterations in Stage~II is at most $2 \ln 2 \log(\frac{4}{\delta}) \tilde{T}$, and since each iteration takes a number of steps $\Ord{\sqrt{\frac{N}{1+\Delta}}}$, the total query complexity of Stages~I and ~II is $$\Ord{\sqrt{\frac{N}{1+\Delta}}\log(1/\delta) \max{\left(8 \log_{3/2}\left(\frac{N}{4\Delta+3}\right), 16 \ln N, 2\ln 2 (19 \Delta +16)\right)} + \log(1/\delta)^2 }.$$ 
The query complexity of \textsc{COMB}($\delta/4$, $S$) follows from Theorem \ref{thm:comb} as $\Ord{\vert S\vert \log(1/\delta)}$. Thus we know that $\vert S\vert \leq 17\Delta+16 = \Ord{1+\Delta}$. Hence, Stage~III incurs a number of queries of $\Ord{(1+\Delta)\log(\frac{1}{\delta})}$. Adding this to the query complexity for Stage~I and ~II yields the result stated.
\end{proof}

\begin{table}[!h]
\begin{center}
\begin{tabular} { c | c | c | c | c }
\hline
Name & Success prob. & Final guarantee & Run time & Theorem\\
\hline \hline
\textsc{PivotQMF} & $\frac{3}{4}$ & rank$(y)$ $\leq 18 \Delta +16$ & $\tOrd{\sqrt{\frac{N}{1+\Delta}}}$ &\ref{thm:1PCQMF} \\
\hline
\textsc{RepeatedPivotQMF} & $1-\delta$ & rank$(y)$ $\leq 18 \Delta +16$ & $\tOrd{\sqrt{\frac{N}{1+\Delta}}}$ &\ref{thm:SRQMF}\\
\hline
\textsc{RobustQMF} & $1-\delta$ & $d(y,y^\ast) \leq 2$ &  
$\tOrd{\sqrt{N(1+\Delta)}}$ 
&\ref{thm:RQMF_1}
\end{tabular}
\caption{Comparison of the main algorithms in this work. Recall $y^\ast$ is the true minimum element of the list, and $\Delta$ is an upper bound on the number of elements within distance $1$ of any particular element. We quote here simplified run time upper bounds and refer to the theorems for the exact results. 
} 
\label{tableSummary}
\end{center}
\end{table}

\section{Application: Sublinear-time Hypothesis Selection} \label{sec:app}
Consider Problem \ref{prob:hypo} of choosing one hypothesis out of a set of candidate hypotheses $\mathcal P$ of size $N$ which is closest to a given unknown distribution $q$.
Closeness is measured by $\ell_1$ distance and the \emph{Scheff\'{e} test} (described later; for more details see ~\cite[Chapter 6]{DevroyeLugosi2001}) is an efficient statistical surrogate for comparing distances. 

Acharya et al.~\cite{acharya2016maximum} emphasize that hypothesis selection can be mapped to the problem of robust minimum finding with a noisy comparator, when the relevant comparator between hypotheses distances to $q$ is the Scheff\'{e} test. In this Section, we show that with this comparator, the guarantees of \textsc{RobustQMF} lead to a quantum speedup also for the problem of hypothesis selection. 

\begin{defn}
[Hypothesis, informal] We define a hypothesis $p$ over a domain $\mathcal D$ to be a probability distribution (discrete/continuous) over some domain $\mathcal D$. For each $x\in \mathcal D$, $p(x)$ is the probability mass/probability density at $x$.
\end{defn}
First we define the Scheff\'{e} set, which will be important to understand the Scheff\'{e} test.
\begin{defn}
[Scheff\'{e} sets]
For all $p_i,p_j \in \mathcal P$, let the corresponding
Scheff\'{e} set be defined as 
$\mathcal{S}_{ij} = \{x \in \mathcal{D}: p_i(x) > p_j(x)\}$. For any hypothesis $p$ let $p(S_{ij})$ be the probability mass of $p$ in the Scheff\'{e} set $S_{ij}$.
\end{defn}

For the Scheff\'{e} test we require the following way to access the the Scheff\'{e} sets and the probability masses in the Scheff\'{e} sets.
This access is convenient and appropriate for example in a setting where the candidate hypotheses are from known families of distributions such as the Gaussian or Poisson distribution \cite{DaskalakisDS14}.
In addition, sampling access to the true distribution is required. 
\begin{assm}[Data access]\label{assm:scheffeaccess}
 For all $i,j\in [N]$ assume the following access.
\begin{itemize}
    \item Query access to Scheff\'{e} sets: For all $\mathcal S_{ij}$ and for all $x\in \mathcal D$, define query access to the Scheff\'{e} sets in the sense that testing if $x\in \mathcal S_{ij}$ takes time denoted by $\eo_{\mathcal S}$.
    \item Query access to probability mass in Scheff\'{e} sets:
    Assume query access to the real number $p_i(\mathcal S_{ij})$ in time $\eo_p$.
    \item Sampling access to unknown distribution: For the unknown distribution $q$, assume the ability to draw samples $x \in \mathcal D$ with probability $q(x)$. One sample is drawn in $\eo_q$. Once the samples are drawn, assume query access to the samples. 
\end{itemize}
\end{assm}
The Scheff\'{e} test (which has the unknown distribution $q$ as an implicit argument) is a statistical surrogate for comparing the following $\ell_1$ distances: $\lVert p_i -q\rVert_1$, $\lVert p_j -q\rVert_1$.

\begin{algorithm}[!htbp]
    \begin{algorithmic}
        \Require 
        \State Access to distributions $p_i \in \mathcal P$ and $p_j \in \mathcal P$ according to Assumption \ref{assm:scheffeaccess},
        \State $\{x_k\}_{k=1}^K$ i.i.d. samples from unknown distribution $q$.
        \State
        \State Compute $\mu_{\mathcal S} = \frac{1}{K} \sum_{k=1}^K \mathbbm 1_{x_k \in \mathcal S_{ij} }$.
        \Ensure If $|p_i(\mathcal{S}_{ij}) - \mu_{\mathcal{S}}| \leq |p_j(\mathcal{S}_{ij}) - \mu_{\mathcal{S}}|$ output $p_i$, otherwise output $p_j$
    \end{algorithmic}
    \caption{Scheff\'{e} test}
    \label{fig:Scheffe}
\end{algorithm}
\begin{lemma}[Run time]
A single Scheff\'e test between two hypotheses has a run time of $\eo_{\rm Samp} + \eo_{\rm Scheffe}$ where $\eo_{\rm Samp} = K \times \eo_q$ and  $\eo_{\rm Scheffe} = K \times (\eo_{\mathcal S}+1) + 2 \times \eo_p$.
\end{lemma}
We now discuss the guarantees of this test. 
When applied to two hypotheses, $p_i,p_j$, and with an unknown distribution $q$, with probability $1-\delta$ the Scheff\'{e} test outputs an index $i \in [2]$ such  that
\begin{equation}\label{eq:scheffeguarantee}
\lVert p_i- q \rVert_1 \leq 3\ \min_{j\in[2]} \lVert p_j- q \rVert_1  + \sqrt{\frac{10 \log{\frac{1}{\delta}}}{K}},
\end{equation}
where $K$ is the number of samples drawn from the unknown hypothesis $q$. Choose $K$ large enough, see below, 
then the error becomes at most $\varepsilon$.
The output of the Scheff\'{e} test satisfies, with probability at least $1-\delta$,
\begin{equation}
    \left\{\begin{array}{ll}
{p_{i}} & {\text { if }\left\|p_{i}-q\right\|_{1}<\frac{1}{3}\left\|p_{j}-q\right\|_{1}} \\
{p_{j}} & {\text { if }\left\|p_{j}-q\right\|_{1}<\frac{1}{3}\left\|p_{i}-q\right\|_{1}} \\
{\text {either $p_i$ or $p_j$}} & {\text { otherwise. }}
\end{array}\right.
\end{equation}
The Scheff\'{e} test acts as a noisy comparator between two hypotheses: it chooses the one that is closer to the unknown distribution, if one is significantly closer than the other. Otherwise there are no guarantees on the output. 
Note that the result of the Scheff\'{e} test is deterministic given the samples, which means that once the samples are fixed, every time the Scheff\'{e} test is invoked, it will output the same result for every pair of hypotheses. 

Minimum finding requires a comparator. Indeed, taking  \(x_{i}=-\log _{3}\left\|p_{i}-q\right\|_{1}\) recovers a noisy comparator of the form given in Equation \ref{eq:noisycomparator}:
\begin{equation}\label{eq:scheffecomparator}
h(i,j) = \left\{\begin{array}{ll}
{\min \left\{x_{i}, x_{j}\right\}} & {\text { if }\left|x_{i}-x_{j}\right|>1} \\
\text {unknown} & \text { otherwise.}
\end{array}\right.
\end{equation}
Similarly to before, define the corresponding oracle. 
\begin{defn} 
[Scheff\'{e} oracle]\label{oracleScheffe}
For $i,j \in [N]$, let
\begin{align} \label{}
O^{(S)}_{i}(j) &=
\begin{cases}
    &1 \quad \text{if $h(i,j)$ outputs $j$ \text{(i.e. h thinks $j<i$)}}\\
    & 0 \quad \text{otherwise}.
\end{cases}
\end{align}
A call to the oracle costs $\eo_{\rm Scheffe}$.
\end{defn}
Using the Scheff\'{e} test to compare pairs of hypotheses,  this oracle  maps directly to the setting of minimum finding with a noisy comparator. 
We continue with introducing the quantum computing framework. We assume the following natural quantization of the classical access Assumption \ref{assm:scheffeaccess}.
\begin{assm}[Quantum access]\label{assm:scheffeaccessquantum}
 For all $i,j\in [N]$ assume the following quantum access.
\begin{itemize}
    \item Quantum access to Scheff\'{e} sets: For all $\mathcal S_{ij}$ and for all $x\in \mathcal D$, define quantum access to the Scheff\'{e} sets in the sense that we are given a unitary circuit such that for fixed $i$ and $x$ we obtain $\ket j \ket 0 \to \ket j \ket {x\in \mathcal S_{ij}}$ with run time $\eo_{\mathcal S}^{(Q)}$.
    \item Quantum access to probability mass in Scheff\'{e} sets:
    Assume quantum access to a unitary such that for fixed $i$ we obtain $\ket j \ket {\bar 0} \to \ket j \ket {p_i(\mathcal S_{ij})}$ in time $\eo_p^{(Q)}$.
    \item Sampling access: Assume the same access to distribution $q$ as in Assumption \ref{assm:scheffeaccess}. The $K$ samples are fixed in advance and accessible via classical queries.
\end{itemize}
\end{assm}
\begin{defn} 
[Scheff\'{e} quantum oracle]\label{oracleScheffequantum}
Let the access as in Assumption \ref{assm:scheffeaccessquantum} be given. For all $i,j \in [N]$, let
$\hat O^{(S)}_{i}(j)$ be the unitary acting as $\ket j \ket 0 \to \ket j \ket {O^{(S)}_{i}(j)}$.
The cost of a call to the oracle shall be denoted by $\eo^{(Q)}_{\rm Scheffe}$.
\end{defn}
\begin{lemma}[Run time] \label{lemma:runtimequantumscheffeoracle}
The Scheff\'e quantum oracle has a run time of $\eo_{\rm Scheffe}^{(Q)} = K \times (\eo_{\mathcal S}^{(Q)}+1) + 2 \times \eo_p^{(Q)}$.
\end{lemma}
\begin{proof}
The basic arithmetic operations performed in the classical Scheff\'e test can be replicated by the equivalent quantum circuits. We take the accuracy to be a constant number of bits (say 64 bits.)
\end{proof}
With these assumptions, we may apply our robust quantum minimum-finding algorithm. 
\begin{theorem}[\textsc{RobustQMF} for hypothesis selection\label{thm:RQMFforhypothesis}]
Given Assumption \ref{assm:1} and the input from Definition \ref{oracleScheffequantum}, there exists a quantum algorithm with expected number of oracle queries
$$\Ord{\sqrt{\frac{N}{1+\Delta}}\tilde{R}\log\left(\frac{1}{\delta}\right)+ (1+\Delta)\log\left(\frac{1}{\delta}\right) }, $$
where $\tilde{R}:=\max{\left(\Ord{\log\left(N\right)},\Ord{\Delta +1}\right)}$,
that with probability at least $1-\delta$ outputs a hypothesis $\hat{p}$ that satisfies
\begin{equation}\label{eq:scheffeperf}
    \lVert \hat{p} - q \rVert_1 \leq 9\min_{p\in \mathcal{P}_{\delta}} \lVert p - q \rVert_1 + 4 \sqrt{\frac{10 \log \frac{\binom{N}{2}}{\delta}}{k}}.
\end{equation}
\end{theorem}
Note that for any chosen accuracy $\varepsilon>0$, the second term on the right hand side of Equation~\ref{eq:scheffeperf} can be made smaller than $\varepsilon$ by taking $k = \Ord{\frac{1}{\varepsilon^2}\log N}$ samples from the unknown distribution.
\begin{proof}
Let $\delta>0$ be a (constant) probability of error, to be set later in the analysis. Applying the Scheff\'e test with $\delta' = \delta/(2\binom{N}{2})$, we get by~\eqref{eq:scheffeguarantee} and a union bound over all $\binom{N}{2}$ pairs of hypotheses that the output of all possible comparisons between two hypotheses $p_{i_1},p_{i_2}$ (leading to output $j(i_1,i_2)$)
\begin{equation}
    \lVert p_{i_{j(i_1,i_2)}}- q \rVert_1 \leq 3 \min_{j\in[2]} \lVert p_{i_j}- q \rVert_1 + \sqrt{\frac{10 \log{\frac{\binom{N}{2}}{\delta}}}{k}},
\end{equation}
with probability at least $1-\binom{N}{2}\delta' = 1-\delta/2$. We can therefore, for the rest of the analysis, condition on the set of $k$ samples satisfying the above (which is the case with probability at least $1-\delta/2$), after which the comparisons induced by the Scheff\'e test are deterministic and meet the guarantee~\eqref{eq:scheffecomparator}.

Since hypothesis selection can be mapped directly to robust minimum finding with the Scheff\'{e} test as the noisy comparator, the quantum algorithm in question is simply \textsc{RobustQMF}. Theorem \ref{thm:RQMF_1} states that if \textsc{RobustQMF} completes successfully, then it will output a $2$-approximation of the minimum. The approximation guarantee then emerges from Equation \ref{eq:scheffeguarantee} and the transformation \(x_{i}=-\log _{3}\left\|p_{i}-q\right\|_{1}\) (for an elaboration, see \cite[Theorem 16]{acharya2016maximum}). 
The query complexity is as stated. 
Note that each query costs $\eo^{(Q)}_{\rm Scheffe}$ from Lemma \ref{lemma:runtimequantumscheffeoracle}.

By a union bound, recalling the simultaneous probability of success of all Scheff\'e estimates (which is $1-\delta/2$ by our previous discussion) and that of the algorithm \textsc{RobustQMF} (again $1-\delta/2$) , the algorithm outputs a hypothesis $\hat{p}$ satisfying our requirements with probability at least $1 - \delta$.
\end{proof}

\section{Discussion and conclusion}


In this paper, we have presented a generalization of the original quantum minimum finding algorithm by D\"{u}rr and H{\o}yer that makes it robust to a comparator that suffers from noise in distinguishing elements. 
We have achieved a quantum speedup for noisy minimum finding and subsequently have applied this speedup to hypothesis selection. Hypothesis selection typically shows run times linear in the number of hypothesis, while using quantum resources we can provide a sublinear time algorithm.

We comment on some pragmatic modifications of our algorithms that might enable them to run faster on a quantum device. Regarding quantum exponential search, we note that an algorithm with similar run time and success probability guarantees called Fixed-point Amplitude Amplification \cite{Yoder_2014} was introduced. This algorithm can in principle replace the search component of the algorithms discussed in this work. Secondly, in \textsc{RobustQMF}, we could implement a quantum counting step just before Stage ~II to accurately estimate the number of elements marked by the noisy comparator as less than $Y_{\rm out}$ with at most a constant run time overhead. Subsequently, this number would be used to decide $T$. This step may decrease the expected run time of Stage~II by a constant factor. 

In fact, our work may lead to a guarantee for the original quantum minimum finding. Note that \textsc{PivotQMF} and the proofs contained therein can be viewed as an analysis technique for the original \textsc{QMF} algorithm in the noisy setting. In Algorithm~\ref{algo:noiseless_qmf}, instead of the value for $T_{\max}$ derived by Ref.~\cite{durr1996quantum}, an increased $T_{\max}$ will suffice to obtain one of the desired low-rank elements with high probability.
It is left for future work to improve the run time of our algorithms by $\log$ factors analogous to the D\"{u}rr-H{\o}yer algorithm and potentially with different assumptions on the noise. 
This investigation will involve a direct computation of $p(N,r)$ from Ref.~\cite{durr1996quantum} in the noisy setting. 

\section{Acknowledgments}\label{sec:ack}
The authors thank Miklos Santha and Stanislav Fort for valuable discussions.

\bibliography{refs} 
\bibliographystyle{acm}

\end{document}